\documentclass[12pt,titlepage]{article}

\pdfoutput=1 
\usepackage{setspace}

\usepackage[T1]{fontenc}
\usepackage{lmodern}
\usepackage[colorlinks=true, pdfstartview=FitV, linkcolor=blue,citecolor=blue, urlcolor=blue]{hyperref}

\usepackage[margin=1.1in]{geometry}

\usepackage{amsmath, amsfonts,mathtools,amsthm,verbatim}
\usepackage[square]{natbib}

\newcommand*\diff{\mathop{}\!\mathrm{d}}

\def\Re{\mathbf{R}} \def\R{\mathbf{R}}

\def\ep{\varepsilon}

\def\al{\alpha}
\def\la{\lambda}
\def\da{\delta}
\def\g{\gamma} 
\def\phi{\varphi}

\def\one{\mathbf{1}}

\DeclareMathOperator{\supp}{supp}
\def\ul{\underline}

\newcommand{\df}[1]{\textit{#1}}
\newcommand{\norm}[1]{\| #1 \|}

\newcommand{\abs}[1]{| #1 | }

\makeatletter
\def\paragraph{\@startsection{paragraph}{4}%
  \z@\z@{-\fontdimen2\font}%
  {\normalfont\bfseries}}
\makeatother

\renewcommand{\df}[1]{\textbf{\textit{#1}}}

\def\cite{\citet}

\newtheorem{theorem}{Theorem}
\newtheorem{corollary}[theorem]{Corollary}
\newtheorem{lemma}[theorem]{Lemma}
\theoremstyle{remark}
\newtheorem*{remark}{Remark}

\begin{document}

\title{Decision theory and the \\ ``almost implies near'' phenomenon.\thanks{We are grateful to Marcelo Gallardo for comments on a previous draft.} }

\author{\begin{tabular}{cc}
Christopher P. Chambers\thanks{Chambers: \href{mailto:Christopher.Chambers@georgetown.edu}{Christopher.Chambers@georgetown.edu}.} & Federico Echenique\thanks{
Echenique: \href{mailto:fede@econ.berkeley.edu}{fede@econ.berkeley.edu}.} \\
{\small Georgetown University} & {\small UC Berkeley} 
\end{tabular}
}

\maketitle

\begin{abstract}
We examine behavioral axioms in decision theory that are satisfied approximately rather than exactly. We demonstrate that in key domains---decisions under risk, uncertainty, and intertemporal choice---behavior that \emph{almost} satisfies an axiom implies the existence of a utility function that is \emph{near} one that adheres to the standard theoretical representation (e.g., expected utility, or exponentially discounted utility). We explicitly quantify the distance between the utility that captures actual behavior and the ideal theoretical utility as a function of the measured deviation from the axiom. This result formally connects two distinct quantitative exercises: measuring empirical deviations from theory and utilizing approximate optimization. Effectively, we show that small deviations from behavioral axioms rationalize the use of standard models as valid approximations.
\end{abstract}

\section{Introduction}

We study behavioral axioms that need not hold exactly, only approximately. In practice, economists routinely overlook small departures from theoretical predictions. This tolerance motivates relaxed versions of the standard axioms in decision theory and behavioral economics.

Across several core domains of decision theory, we show that such a relaxed axiom generally delivers \emph{two} utility representations. One utility function directly tracks observed behavior, encoding the actual deviations from the exact axiom. The other is a ``benchmark'' utility that satisfies the standard, non-relaxed axiom and admits the usual representation. For example expected utility, or exponential discounting. Our main results show that these two utilities must be close to each other, with the distance between them tightly controlled by the size of the axiom’s violation. Put plainly: behavior that \emph{almost} satisfies an axiom generates a utility that lies \emph{near} one that satisfies it exactly.\footnote{The phrase ``almost implies near'' is from \cite{anderson1986almost}, who analyzes the general mathematical phenomenon in which an object that almost satisfies a property guarantees the existence of a nearby object that satisfies it exactly. One of Anderson's applications is to the stability of the Cauchy equation, which is closely related to the techniques we use in some parts of our analysis.}

The standard program in decision theory is to prove that an axiom $A$ is equivalent to a specific functional form of a representing utility $v$. For instance, in the case of choice over risky lotteries, let $A$ be the independence axiom; then $A$ (plus standard auxiliary assumptions) yields an expected utility representation. We instead work with relaxations $A(\ep)$ of $A$, meaning that $A$ holds only ``up to $\ep$.'' Crucially, in our framework, $\ep$ is an observable, empirically meaningful quantity that can, in principle, be measured with data.

Our results say: if the $\ep$-relaxed axiom $A(\ep)$ holds, then there exists a number $\da$ and two suitably normalized utility functions, $u$ and $v$, such that $u$ represents the observed behavior, $v$ is an expected utility representation, and $u$ and $v$ are $\delta$-close. Moreover, $\delta$ and $\ep$ are quantitatively linked: smaller violations of the axiom (smaller $\ep$) guarantee a smaller distance $\delta$ between $u$ and $v$. In fact our main results connect $\da$ and $\ep$ quantitatively. Thus, optimizing the ``ideal'' objective $v$ automatically delivers an approximately optimal choice for the behavioral utility $u$.

Our results only become meaningful once we fix a canonical normalization for the functions $u$ and $v$. Without it, any monotone transformation of these utilities could make the distance between them arbitrarily small or large. To avoid this arbitrariness, we impose a specific normalization that anchors the magnitude of $\da$ and makes its size interpretable.

We see our results as compelling for several distinct reasons.

First, empirical work often interprets the distance between two utilities as a measure of the distance between theory and data. To use a recent example, \cite*{mcgranaghan2024distinguishing} argue that many of the most common rejections of expected utility theory can be rationalized by a relatively small perturbation added to expected utility theory (a similar point is made earlier in \cite{loomes2005modelling}). In particular, some of the most studied versions of the Allais paradox can be rationalized with an additive shock to a standard expected utility function. What has been missing, we think, is a tight connection between relaxing the axioms underlying the theory and the existence of a ``close by'' representation within the theory. In what sense does a small perturbation in expected utility arise from a small deviation from the main behavioral axiom underpinning expected utility theory?

Second, approximate optimization, or $\ep$-maximization, is a common way of describing behavior in complex optimization problems. The idea goes back to \cite{simon1955behavioral}. Instead of assuming or requiring that an objective function achieves an exact optimum, it is common to look for a solution that is not much worse than optimal. Approximate optimization is very common in problems with computational constraints and is popular among computer scientists. It is also common in game theory and learning, though it remains an ad-hoc procedure that can be met with suspicion by economists due to its cardinal nature. 

We connect two quantitative exercises: measuring deviations from axioms and approximating optimization. One exercise measures the degree of consistency between the theory and behavior by the size of the mistake (as in, for example, Quantal Response Equilibrium; see \cite{mckelvey1995quantal}) that is needed to explain the observed behavior. Another exercise is approximate maximization: observed choices that are close to optimal according to the theory. Our paper shows a sort of equivalence between these two views. We connect the approximation in the second exercise to the measured deviation in the first. In particular, our results help ground the cardinal assumptions made whenever an approximation algorithm is used. So our paper provides, we think, a new and hopefully useful perspective on approximate optimization.

Our results pertain to three different choice domains: decisions under risk, decisions under uncertainty, and intertemporal choice. These were chosen to make a point. The domains reflect the most common areas of study in decision theory and behavioral economics. In each case, we obtain results that follow a similar pattern. There is an axiom that is generally viewed as the key property in obtaining the most commonly used class of utility representation. We present a relaxation of the axiom, which can be viewed as the extent to which the axiom is satisfied by the primitive preference relation. The relaxed axiom is then connected to the existence of a utility representation that is close to a rationalization of the preference, while still satisfying the exact version of the axiom. For example, for choice under risk (Section~\ref{sec:EU}) we show that if a preference satisfies an approximate version of the independence axiom, then there is an expected utility function that is close to a utility representation for the preference.

For choice under uncertainty, we work with a streamlined version of the Anscombe–Aumann framework that permits choice over state-contingent monetary payoffs. Our main findings are in Section~\ref{sec:AA}. There, once again, a weakened form of the independence axiom guarantees the existence of two utility representations: one that directly represents the given preference and another that is linear. These two representations are provably close, with the gap precisely controlled by the initial violations of independence. We further show that when the preference is homothetic, the linear representation becomes exact. Since the literature on choice under uncertainty also emphasizes homogeneity and uncertainty aversion, we derive approximate analogues of these properties and prove the associated approximate representation theorems.

We then apply our general results to two prominent models: maxmin expected utility and the smooth ambiguity model. The maxmin model is homothetic, and we demonstrate that our main theorems cannot be used to approximate it with a subjective expected utility representation. Turning to the smooth ambiguity model, we study specifications in which departures from ambiguity neutrality fade away as the monetary stakes of acts grow large. For these particular versions of the smooth model, we obtain a global approximation via a subjective expected utility representation. The message is not that one should expect ambiguity neutrality for high-stakes acts, but rather that behavior in large-stakes domains is what ultimately determines whether an expected utility approximation exists.

Finally, in the case of intertemporal choice, we look at the choice over so-called dated rewards (or delayed rewards). This environment is common in experimental tests of the theory. Again we provide three approximation results assuming versions of the stationary axiom. The results are in Section~\ref{sec:dated}, where we present two different treatments depending on whether we assume discrete or continuous time.

In two of our results, we show that it is possible to recover an exact representation result using a relaxed version of the relevant axiom. When we assume a relaxed version of independence together with homotheticity, we obtain an exact linear representation (Theorem~\ref{thm:AA}). In our discussion of intertemporal choice, a relaxed version of stationarity, when restricted to preferences with strictly decreasing discount factors, implies an exact geometric discounting representation (Theorem~\ref{thm:exp2}).

In terms of the relevant techniques, we have relied heavily on the theory of stability of functional equations for the findings in Section~\ref{sec:AA} and~\ref{sec:dated}. An exposition of this theory can be found, for example, in \cite*{hyers2012stability}.

\paragraph{Related literature.}

The closest work to ours examines the distance between theory and behavior through axioms that are essentially tests. They are conditions that a finite dataset may or may not satisfy. The papers by \cite{declippel2023relaxed} and  \cite*{echenique2023approximate}, in particular, work from the first-order conditions of the optimization problem that agents would be solving if they behaved as the theory predicts. They then relax these conditions and consider a model to be consistent with the theory if the relaxation needed is not too large. The focus of these papers is on developing an empirically tractable methodology, with an emphasis on a particular domain of economic theory. The starting point is the theory, not the axioms. In contrast, we have started from a given axiom, understood as the substantive primitive behavior that the theory seeks to capture, and considered the consequences of relaxing the axioms. Instead of developing a test for a particular theory, we aim to make a general methodological point. The point is that decision theory does not need to be black and white. It admits some shades of gray. An axiom can hold approximately, and this can have some interesting consequences. We stress how one can connect the degree to which an axiom is relaxed to the approximation guarantee in terms of an objective function to optimize.

There is a more distant and more applied literature that looks at the discrepancy between data and rationality broadly defined: the fundamental ideas are laid out in \cite{afriat1973system},\cite{VARIAN1990125}, and \cite{houtman1985determining}. This literature does not establish a formal connection between ``almost'' and ``near.'' Instead, it relies on the intuition that large measured deviations from rationality make the data difficult to reconcile with the theory.

\section{Motivation.} 

We discuss a version of the Allais paradox of expected utility theory, and how it is related to the approximate satisfaction of the axioms underlying the theory.\footnote{A discussion of different violations of expected utility can be found in \cite{machina87}. A  very extensive discussion of the earlier empirical literature may be found in \cite{camerer95}. A recent meta study focusing on a particular patter of violations of expected utility can be found in \cite{blavatskyy2023common}.} 
Consider the following lotteries over monetary payments:

\begin{itemize}
\item Lottery $A$ yields \$4000 with probability 80\% and \$0 with probability 20\%.
\item Lottery $B$ yields \$3000 with probability 100\%.
\item Lottery $C$ yields \$4000 with probability 20\% and \$0 with probability 80\%.
\item Lottery $D$ yields \$3000 with probability 25\% and \$0 with probability 75\%.
\end{itemize}

These lotteries are variants of an example discussed by \cite{allais1953comportement}. A common pattern of choice is to prefer lottery $B$ over $A$, and $C$ over $D$.  \cite{kanhneman1979prospect} carried out an experiment in which they asked respondents to make these choices (using the same numbers, but expressed in pounds, the Israeli currency of the time). In their experiment, 82\% of respondents preferred $B$ over $A$, while 83\% preferred $C$ over $D$. This pattern of choices is, however, inconsistent with expected utility theory. This is easy to see directly, but it is also instructive to note that it violates the independence axiom, which states that if $p,q,r$ are lotteries and $p\succeq q$, then for any $\la\in (0,1)$, $\la p + (1-\la) r\succeq \la q+(1-\la) r$. Here, when $B$ is preferred to $A$, the axiom requires that $D$ is preferred to $C$ because $C$ may be obtained as $0.25\times A + 0.75\times 0$ and $D$ as $0.25\times B+0.75\times 0$. The pattern of choices discovered by Kahneman and Tversky is often called the common ratio effect.

The common ratio effect and other versions of the Allais paradox can be rationalized using non-expected utility; perhaps the best known being Prospect Theory. To provide an example, we shall use Cumulative Prospect Theory with the functional forms proposed by \cite{kanhneman1979prospect} and the parameter estimates of \cite{wu1996curvature}. This means evaluating a lottery with two-point support, yielding $x$ with probability $p$ and $z$ with the complementary probability, by the utility function $g(p)w(x)+(1-g(p))w(z)$, where $w(x)=x^{0.54}$ and 
$$g(p)=p^{0.74}[p^{0.74}+(1-p)^{0.74}]^{-1/0.74}.$$ Such a representation  can rationalize the pattern of choices in the Allais paradox: the utility of $B$ is larger than the utility of $A$, while the utility of $C$ is larger than the utility of $D$. Using these numbers, we can predict that the lottery $D'$ obtained as $0.27\times B+0.73\times 0$ would actually be preferred to lottery $C$. So after we amend one of the lotteries, the comparisons of $A$ to $B$ and of $C$ to $D'$ do not exhibit the Allais pattern. Compare $D'$ to $D$. The difference is quite small.

The point is that such a small difference in probabilities suffices to ``amend'' the independence axiom, at least as far as these four lotteries go. In other words, by adjusting the mixture probabilities by $\ep=0.02$ in the requirements of the independence axiom, we can reconcile the non-expected utility behavior encoded in cumulative prospect theory with an approximate version of the independence axiom. Of course, this is constrained to the choices in the Allais paradox, not the global behavior of the theory. Our results will imply that whenever an adjustment to the axiom is satisfied by the observed choices, there is an expected utility functional form that approximates a utility rationalizing the choices.  

\begin{figure}
    \centering
    \includegraphics[width=0.5\linewidth]{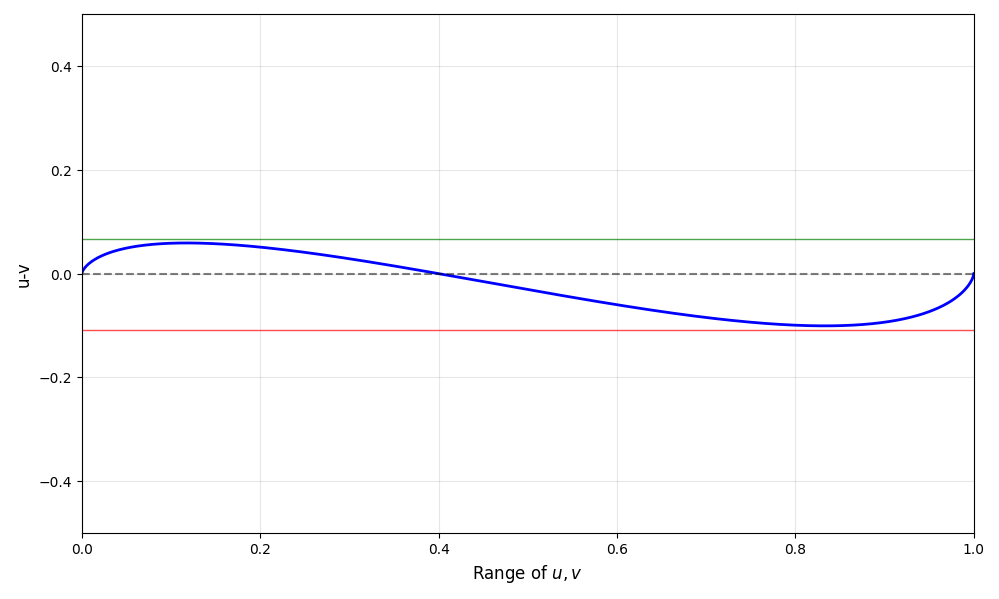}
    \caption{Utility differences: Cumulative Prospect Theory and Expected Utility.}
    \label{fig:prospectthy}
\end{figure}

Figure~\ref{fig:prospectthy} illustrates the global differences in terms of normalized utilities between the cumulative prospect theory in the specification estimated by \cite{wu1996curvature} and an expected utility function that achieves the same values at the best and worst prizes: $0$ and $1$. We focus on simple lotteries that offer a chance $p$ of winning a best prize versus with utility $1$ and a worst prize with utility $0$. Under expected utility theory, the value of such a lottery is simply the probability itself, $v(p)=p$. Under Cumulative Prospect Theory (using the specification estimated by \cite{wu1996curvature}), this probability is distorted by a non-linear weighting function, $u(p)=w(p)$. The figure plots the difference between these two valuations, $u(p) - v(p)$, effectively showing the size of the deviation from the standard axiom. As the graph demonstrates, while the deviation is real, the two utility evaluations remain quite close, differing by at most $0.1$.

\section{Model and results}

We proceed to revisit some of the most popular models in decision theory and propose an ``almost implies near'' approach to the best-known axioms and utility representations in these domains.

\subsection{Basic definitions and notation}
    
A binary relation $\succeq$ on a set $X$ is a \df{weak order} if it is complete and transitive. All preference relations are assumed to be weak orders. If $x\in X$, then the set  $\{z\in X \mid z\succeq x\}$ is the \df{upper contour set} of $\succeq$ at $x$, and $\{z\in X \mid x\succeq z\}$ is the \df{lower contour set} of $\succeq$ at $x$.  When $X$ is a topological space and the upper and lower contour sets are closed, we say that $\succeq$ is \df{continuous}.

Let $x,y\in\Re^d$. We write $x\geq y$ when $x_i\geq y_i$, $i=1,\ldots,d$; $x>y$ when $x\geq y$ and $x\neq y$; and $x\gg y$ when $x_i> y_i$, $i=1,\ldots,d$. The set of $\{x\in\Re^d \mid x\geq 0\}$ is denoted by $\Re^d_+$, and the set of $x\gg 0$ by $\Re^d_{++}$. The vector $(1,\ldots,1)$ in $\Re^d$ is denoted by $\one$. The norm $\norm{x}_{\infty}$, for $x\in\Re^d$, equals $\max \{\abs{x_i} \mid 1\leq i\leq d \}$.

The \df{simplex} in $\Re^d$ is the set $\Delta=\{p\in\Re^d_+ \mid \sum_{i=1}^d p_i=1 \}$. The \df{support} of $p\in\Delta$, denoted by $\supp(p)$, is the set of $i\in \{1,\ldots,d\}$ with $p_i>0$. 

A \df{preference} on $\Re^d_+$ is a binary relation that is complete and transitive. These are also called weak orders. A function $u:\Re^d_+\to\Re$ is a \df{representation} of $\succeq$ if $x\succeq y$ iff $u(x)\geq u(y)$. A representation $u$ is \df{normalized} if $u(c\one)=c$ for all $c\in\Re_+$.

\subsection{Expected Utility Theory}\label{sec:EU}

Our first results follow the motivation very closely. We consider decision-making under risk, the realm of the von Neumann--Morgenstern theory of expected utility for objective lotteries. In decision-making under risk, the primitive choice space is the set of all lotteries over a finite set of $d+1$ prizes. 

Let $\Delta=\{p\in\Re^{d+1}_+ \mid  \sum_{i=1}^{d+1}p_i =1\}$ be the $d$-dimensional simplex in $\Re^{d+1}$. The elements of $\Delta$ are called \df{lotteries}. The lotteries that assign probability one to a prize are called \df{degenerate}.

Let $\succeq$ be a binary relation on $\Delta$, and consider the following axioms on $\succeq$, where $\ep> 0$.

\begin{enumerate}
    \item Continuous weak order: $\succeq$ is continuous, complete and transitive. 
    \item Extremality: Let $\bar p$ and $\ul p$ be, respectively, best and worst elements of $\Delta$ according to $\succeq$. That is, $\bar p\succeq p \succeq \ul p$ for all $p\in \Delta$. We assume that $\bar p$ and $\ul p$ are degenerate lotteries.
    \item Monotonicity: If $\la,\la'\in [0,1]$ with $\la>\la'$ then $\la\bar p+(1-\la)\ul p\succ \la'\bar p+(1-\la')\ul p$.
    \item $\ep$-Reduction of compound lotteries: If
\begin{align*}
p_0 &\sim \al_0 \bar p + (1-\al_0)\ul p \\
p_1 &\sim \al_1 \bar p + (1-\al_1)\ul p \\
\la p_0 + (1-\la) p_1 & \sim \al \bar p + (1-\al) \ul p
\end{align*}then 
\[
\abs{\al - (\la \al_0 + (1-\la) \al_1)}<\ep.
\]
\end{enumerate}

Most of the properties in these axioms should be familiar or natural. That a preference is a continuous weak order is the requirement for Debreu's utility representation theorem. Continuity ensures that there exists a continuous utility representation of the preference. Extremality means that there are degenerate lotteries at the top and bottom of the preferences in question; this assumption is certainly natural when we consider monetary lotteries. The next axiom, Monotonicity, is also common and probably uncontroversial for monetary lotteries.

The last axiom is our contribution. It is a relaxed version of a reduction of the compound lotteries axiom. It states that a failure to reduce compound lottery $\la p_0 + (1-\la) p_1$ to its constituent lotteries can be amended by a change in lottery weights that is bounded by $\ep$. The parameter $\ep$ quantifies the extent to which the strict version of the reduction of compound lotteries is ``almost'' satisfied. 

The claims of indifference in the axiom are easy to implement experimentally since experimental methodology often uses multiple price lists (or the BDM mechanism) to elicit indifferences. 

Next, we claim that the axioms imply the existence of an expected utility function (an affine utility) that is close to a utility representation for $\succeq$. A form of converse statement is also true.

\begin{theorem}\label{thm:EU}
    If $\succeq$ satisfies the axioms, then there is a pair of functions $(u,\ell):\Delta\to [0,1]^2$ such that \begin{enumerate}
        \item $u$ represents $\succeq$ and $u(\Delta)=\ell(\Delta)=[0,1]$.
        \item $\ell$ is affine.
        \item $\abs{u(p)-\ell(p)}< (\abs{\supp(p)}-1) \ep$ for all nondegenerate $p\in\Delta$ and $u(\delta_i)=\ell(\delta_i)$ for all degenerate lotteries $\delta_i$, $i=1,\ldots,d+1$. In particular, $\norm{u-\ell}_{\infty}< d \ep$.
        \item Given $u$, $\ell$ is unique among the functions that satisfy these properties.
    \end{enumerate} 
    Conversely, if the pair $(u,\ell)$ satisfies these properties with $\norm{u-\ell}_{\infty}<\ep$, then $\succeq$ satisfies $4\ep$--Reduction of compound lotteries.
\end{theorem}

The proof of Theorem~\ref{thm:EU} and all results in the paper is in Section~\ref{sec:proofs}.

Note that Theorem~\ref{thm:EU} states that the functions $u$ and $\ell$ share a normalization, meaning that they are onto $[0,1]$ and coincide on the degenerate lotteries. The magnitude $\ep$ has an objective and observable meaning in the axiom of $\ep$-reduction of compound lotteries. It is, in principle, measurable experimentally. The translation of $\ep$ into a ``near'' guarantee requires some cardinal meaning that is ensured by the normalization in the theorem.

Second, we should emphasize that the approximation guarantee is better for the comparison of lotteries that have small support. This is the case in many experimental settings. In experiments, there is also a common finding that probabilities that are close to certainty are particularly distorted (\cite{kanhneman1979prospect}). So it is interesting that $u$ and $\ell$ coincide on degenerate lotteries. By continuity, then, the approximation guarantee should be better for lotteries that are close to being degenerate, or close to certainty (of course, to establish this quantitatively would require much stronger assumptions than what we have put in place here).

Theorem~\ref{thm:EU} relies on an approximate version of the axiom of reduction of compound lotteries, but it may be more natural to consider relaxing the independence axiom. In the development of expected utility theory, the independence axiom has played a very important role (\cite{fishburn89}). Very similar ideas to those that we used in Theorem~\ref{thm:EU} can be applied, although the resulting approximation guarantee is somewhat weaker. 

Consider then the following axiom:

$\ep$-Independence: if $p\sim q$, $\al\in [0,1]$, and $r\in \Delta$ then there is $\al'\in [0,1]$ with $|\al-\al'| <\ep$ such that $\al p + (1-\al) r\sim \al' q + (1-\al') r$.

\begin{theorem}\label{thm:EU2}
    Let $\succeq$ be a continuous weak order that satisfies $\ep$-independence. Then there exists a pair of functions $(u,\ell):\Delta\to\Re^2$ such that:
    \begin{enumerate}
    \item $u$ represents $\succeq$, $u(\Delta)=\ell(\Delta)=[0,1]$, and  $u(\delta_i)=\ell(\delta_i)$ for all degenerate lotteries $\delta_i$, $i=1,\ldots,d+1$. 
        \item $\ell$ is affine.
        \item $\norm{u-\ell}_{\infty}<(d+1)^2 \ep$.
    \end{enumerate}
\end{theorem}

It is possible that Theorem~\ref{thm:EU2} admits a converse statement, like the one in Theorem~\ref{thm:EU}. At this point, this is an open question.

%%%%%%%%%%%%%%%%%%%%%%%%%
\section{Choice under uncertainty: the Anscombe-Aumann model}\label{sec:AA}

The previous results dealt with risk in the form of objective lotteries. Another important topic in decision theory is uncertainty, which refers to situations that do not come with known probabilistic assessments. We turn then to the standard workhorse model of subjective uncertainty in modern decision theory, the Anscombe-Aumann model. There is a finite state-space of cardinality $d$, and the objects of choice are state-contingent objective lotteries. In the standard model, the von Neumann--Morgenstern axioms are imposed on objective lotteries so that an expected utility representation over state-independent objective lotteries is ensured. Here, we do not bother to spell out this step; instead, we work with state-contingent utility. Alternatively, our model can be interpreted as being about state-contingent monetary payments for an agent who is risk neutral over objective monetary lotteries. 

In sum, the consumption space is $\Re^d_+$. Each element of $\Re^d_+$ is a state-contingent monetary payment. As a notational convention, for any real number $c\in\Re_+$,  the vector $c\one\in\Re^d_+$ that pays $c$ in each state is also denoted as $c$. So when we write $c$ for a $d$-dimensional vector, we mean $c\one$.

The version of the Anscombe-Aumann model that we adopt is the same as the standard model of consumer choice used in demand theory. So, it is perfectly possible to interpret our results as providing conditions for the approximate optimization of linear utility (meaning that the goods in question are perfect substitutes), quasi-concave utility, and so on.

We consider the following axioms on a preference relation. 

\begin{enumerate}
    \item Monotonicity: $x\leq y$ implies $y\succeq x$, and $x\ll y$ implies $y\succ x$.
\item Homotheticity:  $x\succeq y$ and $\la>0$ imply $\la x \succeq \la y$.
\item\label{ax:phiapprox} Given a function $\phi:X\times X\rightarrow \Re_+$; $\phi$-approximate independence:  for any $x,y\in X$, $\bar c_x,\ul c_x\in\Re_+$ and $\bar c_y,\ul c_y\in\Re_+$ such that
  $\bar c_x \succeq x \succeq \ul c_x$ and
  $\bar c_y \succeq y \succeq \ul c_y$, it holds that
    \[
  \frac{1}{2}\bar c_x + \frac{1}{2}\bar c_y + \phi(x,y)
  \succeq \frac{1}{2}x+\frac{1}{2}y
  \succeq \frac{1}{2}\ul c_x+\frac{1}{2}\ul c_y - \phi(x,y)
  \]
\end{enumerate}

The independence axiom for the Anscombe-Aumann model provides a linear representation. In our $\phi$-approximate independence axiom, the function $\phi$ measures the extent to which independence is violated. Theorem~\ref{thm:AA} below concerns situations when the magnitude of these violations grows sub-linearly with the amount of money involved, in the sense that $\phi(\al x,\al y)$ is $o(\al)$. This means that the violations of independence that one should be concerned about occur for scaled-up versions of any two acts (see Corollary~\ref{cor:smoothamb} below).

The meaning of the other axioms should be familiar. Monotonicity means that money is desirable in every state. Homotheticity states that all preferences are maintained by scaling; it is a common axiom used, for example, in the max-min representation result of \cite{gilboa1989maxmin}. Homotheticity is particularly useful here because it provides an exact, rather than an approximate, linear representation.

\begin{theorem}\label{thm:AA}
    Let $\succeq$ satisfy continuity, monotonicity, and $\phi$-approximate independence. Suppose that $\sum_{i=0}^{n} 2^{-i}\phi(2^{i}x,2^{i}y)\leq \Theta$ for all $n\geq 1$ and $x,y\in X$.
    \begin{enumerate}
        \item There exists $u:X\to \Re$ and a linear function $v:\Re^d\to\Re$ such that $u$ is a normalized representation of $\succeq$ and $\norm{u-v}_{\infty}\leq\Theta$.
        \item If $\succeq$ is homothetic, then $u=v$.
    \end{enumerate}
\end{theorem}

\begin{remark}
The function $v$ in Theorem~\ref{thm:AA} can be interpreted as a subjective expected utility function. Indeed, 
in the proof of Theorem~\ref{thm:AA}, it turns out that $u(c\one)=c$ and \[
v(\mathbf{1}) = \lim_{n\to\infty} 2^{-n}u(2^n \mathbf{1}) = \lim_{n\to \infty} 2^{-n}(2^n) = 1.
\]
By defining $p_i=v(e_i)$, we can ensure that $1=v(\mathbf{1}) = \sum v(e_i) = \sum p_i$. Then the $p_i$ may be interpreted as subjective probabilities and we have that 
\[
v(x) = v(\sum_{i=1}^d e_i x_i)= \sum_{i=1}^d p_i x_i
\]
\end{remark}

The hypotheses in Theorem~\ref{thm:AA} are technical, but their meaning is simple, and the theorem has clear implications in different special cases. First, note that $\Theta$ captures the extent to which independence is violated and is, in principle, observable. Second, the message in the theorem is that violations of independence for large acts are particularly important in obtaining an approximation by means of a linear representation. This fact is formalized in Corollary~\ref{cor:smoothamb} below, where we consider the smooth ambiguity model of \cite{klibanoff2005}.

Note that for homothetic preferences, the approximation has no quantitative bite. Regardless of how large the deviations from independence, as measured by $\Theta$, are, we obtain an exact representation by means of a linear function. In other words, the summability hypothesis in the theorem, together with homotheticity, suffices. Below, we discuss the special case of maxmin expected utility preferences, which are always homothetic, and show why they do not satisfy the summability condition in the theorem.

In the absence of homotheticity, perhaps the simplest implication of the theorem is for the special case of a uniform upper bound on $\phi$. Our next result records the implications of bounding the values of $\phi$.

\begin{corollary} Let $\succeq$ satisfy continuity, monotonicity, and $\phi$-approximate independence.  Suppose that $\phi(x,y)\leq\ep$ for all $x$ and $y$. Then there are functions $u$ and $v$ as in the theorem with  $\norm{v-u}_{\infty}\leq 2\ep$.
  \end{corollary}

Theorem~\ref{thm:AA} establishes that when we impose Homotheticity, we obtain an exact linear representation. It is therefore natural to consider the case of approximate homotheticity.

Let $\bar\phi,\ul\phi : \Re^d_{++}\times \Re^d_{++}\times\Re_+\to\Re$ be two functions that serve as the parameters in the axiom. Suppose that $x\mapsto \bar\phi(x,y,\la)$ is weakly monotonically decreasing, and that $y\mapsto \ul\phi(x,y,\la)$ is weakly monotonically increasing.

A preference $\succeq$ satisfies $(\bar\phi,\ul\phi)$-homogeneity if, when $\la>0$,
\[
x\succeq y \text{ then } \la x + \bar\phi(x,y,\la) \succeq \la y \text{ and } \la x\succeq \la y - \ul\phi(x,y,\la)
\]

\begin{theorem}\label{thm:approxhomog}
    Suppose that $\succeq$ on $\Re^d_{++}$ satisfies continuity, monotonicity, and $(\bar\phi,\ul\phi)$-homogeneity. Suppose that there exists some $\eta>0$ with the property that 
    \[ \sum_{n=1}^{\infty} \frac{\ul\phi(\eta^n x,\eta^n y,\la)}{\eta^n}\leq\Theta
    \text{ and }  \sum_{n=1}^{\infty} \frac{\bar\phi(\eta^n x,\eta^n y,\la)}{\eta^n}\leq \Theta.\]
    Then there is $u,v:\Re^d_{++}\to\Re$ such that $u$ is a normalized representation of $\succeq$, $v$ is homogeneous, and $\norm{u-v}_{\infty}<2\Theta$.
\end{theorem}

The meaning of the hypothesis in the theorem is, essentially, that 
        \[ \lim_{n\to\infty} \frac{\ul\phi(\eta^n x,\eta^n y,\la)}{\eta^n}
        =  \lim_{n\to\infty} \frac{\bar\phi(\eta^n x,\eta^n y,\la)}{\eta^n}=0,\] and that this convergence happens fast enough.
The functions $\bar\phi,\ul\phi : \Re^d_{++}\times \Re^d_{++}\times\Re_+\to\Re$ measure the deviations from homogeneity in the axiom.  The theorem says, then, that to have any bite, deviations from homotheticity of preferences have to occur for large enough acts. For example, if $\bar\phi$ and $\ul\phi$ are bounded by a constant $M$, then for an arbitrarily large $\eta$, we can set $\Theta=\frac{M}{\eta-1}$.  Theorem~\ref{thm:approxhomog} then provides the existence of an arbitrarily close homogeneous approximation.

Finally, we turn to convexity, another common property in the literature on choice under uncertainty. In this context, convexity was introduced by \cite{schmeidler89} and \cite{gilboa1989maxmin} as a notion of uncertainty aversion. If we interpret the model as consumption or demand theory, convexity is also an important property. Convexity is key, for example, in guaranteeing the existence of competitive equilibria and in the study of their welfare properties. 

The ordinal notion of uncertainty aversion corresponds to a quasi-concave utility representation.  We propose a relaxed notion of convexity, or uncertainty aversion.

A preference $\succeq$ satisfies $\ep$-uncertainty aversion if $x\succeq c$ and $y\succeq c$ imply that $\la x + (1-\la) y \succeq c-\ep$ for any $\la\in (0,1)$.

\begin{theorem}\label{thm:quasiquasicon}
    If $\succeq$ is continuous, monotonic, and $\ep$-uncertainty averse, then there exists a normalized utility representation $u:\Re^d_+\to\Re$ and a quasi-convex function $v:\Re^d_+\to\R$ such that $\norm{u-v}_{\infty}\leq d\ep$.
\end{theorem}

To prove the theorem, $v$ is obtained using the convex hull of the preferences' upper contour sets. This idea was first used by \cite{starr1969quasi}. The bound in Theorem~\ref{thm:quasiquasicon} follows simply from a version of Carathéodory's theorem.

\subsection*{Application: Maxmin expected utility and smooth ambiguity aversion}

To clarify the message in Theorem~\ref{thm:AA}, we consider two popular models of choice under uncertainty, the  maxmin expected utility (MEU, \cite{gilboa1989maxmin}) and the smooth ambiguity model of \cite{klibanoff2005}. The MEU model is homothetic, which, given Theorem~\ref{thm:AA}, means that if the summability hypothesis of the theorem were satisfied, it would admit a subjective expected utility representation. It should come as no surprise that this hypothesis is not satisfied; however, we flesh this out here. For the smooth ambiguity model, we work out some examples where the hypothesis of the theorem is actually satisfied. In particular, these examples show that the violations of independence that induce substantial deviations from expected utility involve large-stakes acts.

Suppose then that $\succeq$ is represented by the utility function $u(x) = \min \{\pi \cdot x \mid \pi \in \Pi\}$, where $\Pi\subseteq \Delta$ is a non-degenerate set of priors.
The $\phi$-approximate independence axiom then requires that:
\[
    \frac{1}{2}u(x) + \frac{1}{2}u(y) + \phi(x,y) \geq u\left(\frac{x+y}{2}\right).
\]
Note that the other inequality in the axiom holds as long as $\phi\geq 0$ because $u$ is concave.

Solving for $\phi$, we obtain that 
\[
    \phi(x,y) \geq u\left(\frac{x+y}{2}\right) - \frac{1}{2}u(x) - \frac{1}{2}u(y).
\]
Since $u$ is homogeneous of degree 1, $u(\frac{x+y}{2}) = \frac{1}{2}u(x+y)$. Thus:
\[
    \phi(x,y) \geq \frac{1}{2} [ u(x+y) - u(x) - u(y) ].
\]
If we scale the acts by $\alpha > 0$:
\[
    \phi(\alpha x, \alpha y) \approx \frac{1}{2} [ u(\alpha(x+y)) - u(\alpha x) - u(\alpha y) ]
    = \frac{\alpha}{2} [ u(x+y) - u(x) - u(y) ] = \alpha \phi(x,y).
\]
Since $\phi(\alpha x, \alpha y)$ is linear in $\alpha$, it is not $o(\alpha)$, and the series $\sum 2^{-n}\phi(2^n x, 2^n y)$ diverges. The hypotheses of Theorem~\ref{thm:AA} are not satisfied

Unlike in the case of MEU, one can identify conditions on the smooth ambiguity model of \cite{klibanoff2005} under which it can be approximated by a subjective expected utility representation. To this purpose, consider a utility function of the smooth ambiguity type:
    \[
        u(x) = \int_{\Delta} f\left( \sum_{i=1}^d x_i p_i \right) \diff \mu(p),
    \]
    where $\mu$ is a subjective prior over the set of probability measures $\Delta$, and $f: \Re \to \Re$ is a function that captures departures from subjective expected utility. Note that when $f$ is affine, we obtain a subjective expected utility representation. When $f$ is concave, the representation assumes ambiguity aversion. When $f$ is convex, it assumes ambiguity loving.
    
A sufficient condition for $\phi$-approximate independence is that $\phi$ bounds the $\mu$-weighted average of the degree to which $f$ is not linear. Specifically:
\begin{align*}
    \phi(x,y) & \geq u\left(\frac{x+y}{2}\right) - \frac{1}{2}u(x) - \frac{1}{2}u(y) \\
    & =\int_{\Delta} \left[ f\left( p \cdot \frac{x+y}{2} \right) - \frac{1}{2}f(p \cdot x) - \frac{1}{2}f(p \cdot y) \right] \diff \mu(p).
\end{align*}
    To satisfy the convergence hypothesis of Theorem~\ref{thm:AA}, this cumulative error must grow sub-linearly. 

In the following result (a corollary of Theorem~\ref{thm:AA}), we examine two particular cases of the function $f$. These cases possess two appealing features: one function is concave, and the other is convex; yet both become linear for sufficiently large acts. This helps clarify our point that the only violations of independence that are important are those arising from large acts.

\begin{corollary}\label{cor:smoothamb} Let the function $f$ be either $f(z)=\sqrt{1+z^2}$ or $f(z) = z - e^{-z}$ for $z \geq 0$. Define the mean prior $\bar{p} = \int p \, d\mu$. Let $v$ be the subjective expected utility function with prior equal to $\bar p$:
    \[
        v(x) = \sum_{i=1}^d x_i \bar{p}_i.
    \]
Then, 
    \[
       \norm{u-v}_{\infty} \leq 1
    \]
\end{corollary}

The point of the corollary is to illustrate the role of deviations from ambiguity neutrality for large acts. The two cases covered by the theorem involve representations that become asymptotically ambiguity neutral. The first example, $f(z)=\sqrt{1+z^2}$, is strictly convex and, therefore, features a decision maker who loves ambiguity. The second example, $f(z)=z-e^z$, is strictly concave; thus, it assumes ambiguity aversion. In both cases we have that $\lim_{z\to\infty}f(z)/z = 1$. Hence, for large acts, the two representations approach ambiguity neutrality. 

As a consequence, the corollary shows that we obtain a uniform bound for the difference between the smooth ambiguity representation and a ``natural'' subjective expected utility analogue: the expected utility representation that assumes a mean prior. 

\section{Dated rewards}\label{sec:dated}

We now turn to intertemporal choice. We use the model of dated rewards introduced by \cite{lancaster63} and \cite{fishrubin89} (see \cite{rubinstein2003economics}, \cite{OK2007214} and \cite{dejarnette2020time} for more recent applications). This model fits the experimental designs used in the empirical literature: see, for example, \cite{THALER1981201} for one of the first papers testing exponential discounting, and \cite{shane2002} for a review of the literature.

The primitives are a set $M\subseteq\Re$ of \df{monetary rewards}, and a set $T\subseteq\Re_+$ of possible \df{dates}. A \df{dated reward} is a pair $(X,t)\in M\times T$. Given is a primitive preference relation $\succeq$ on $M\times T$. 

\subsection{Discrete time}\label{sec:disctime}
Suppose that $T=\{0,1,2,\ldots\}$ and that $M=\Re_+$.

We restrict attention to those preference relations that have an associated \df{discount factor} $d:T\to \Re_{++}$ whereby $(X,t)\succeq (Y,s)$ iff $X\cdot d(t)\geq Y\cdot d(s)$. Such preferences have a \df{discount factor representation}. Without loss, we may normalize the discount factors so that $d(0)=1$. We collect preferences with a discount factor representation in a class of preferences that we denote by $\mathcal{D}$.

The main axiom of interest is the following: A preference relation
$\succeq$ satisfies $\delta$-stationarity, for $\da:M\times T\times T\to\Re_+$ if, whenever 
\[
(X,t)\sim (Y,s+t),
\]
then 
\[
(X,0)\sim (Y\da_X(s,t),s)
\]

The usual stationarity axiom (\cite{koopmans1960stationary}) obtains when $\da_X(s,t)=1$ for all $X$, $t$, and $s$. The axiom in that case says that, if we interpret the time $s$ as a \textit{delay}, then the effect of a delay on the comparison of two dated rewards should be independent of the starting point $t$ for the comparison. The stationarity axiom implies that a discount factor must take the exponential form that is standard in economics.

In general, $\da_X$ records the extent to which stationarity is violated when we anchor comparisons by the monetary quantity $x$ and vary the time $t$ and the delay $s$. 

It turns out to be useful to linearize the model by taking logarithms. Define then $\psi_x(s,t)=|\log \da_x(s,t)|$. The difference between the discount factor $d(t)$ and exponential discounting $\gamma^t$ is  evaluated according to the size of $\log(d(t)/\gamma^t)$.

\begin{theorem}\label{thm:exp}
Let $\succeq\in \mathcal{D}$ with the associated discount factor $d$.
    Suppose that, for some $x\in M$, $\sum_{i=0}^{\infty} 2^{-(i+1)}\psi_x(2^i t,2^i s)\eqqcolon\Theta$ is a convergent series. 
    Then there exists a $\gamma>0$ with 
    $$\sup\{\abs{\log \frac{d(t)}{\gamma^t}} :t\in T\} \leq \Theta. $$
\end{theorem}

Theorem~\ref{thm:exp} connects the degree of violation of the stationarity axiom with the existence of a nearby exponential representation for the discount factor. The magnitude $\Theta$ adds up violations and discounts them by $1/2$. This magnitude (which, in principle, depends on the monetary quantity $x$ chosen as the anchor in $\psi_x$), then provides the desired approximation guarantee. 

In the special case when $\da_X$ is uniformly bounded and presumably small (meaning close to 1), we obtain a sharper result, stated here as:

\begin{corollary}
   Suppose that there is $\ep>0$ and $X\in M$ so that  $\da_X(t,s)\in (\frac{1}{1+\ep},1+\ep)$. Then there exists a $\gamma>0$ with 
    $$\sup\{\abs{\log \frac{d(t)}{\gamma^t}}  \mid t\in T\} \leq \ep $$
\end{corollary}

\subsection{Strictly decreasing discount factors}

Let $\bar{\mathcal{D}}$ comprise the preferences $\succeq\in D$ with a discount factor that is strictly monotonically decreasing and vanishes for large enough $t$. Here we allow for negative monetary payments, and thus assume that $M=\Re$.

For the preference in $\bar{\mathcal{D}}$, we entertain a different relaxation of the stationarity axiom. It measures deviations from stationarity relative to the delays involved. 

A preference relation $\succeq$ satisfies the $(X,\Theta)$-stationarity axiom if, for $Y,Z$ and $ W$ with
\[(X,t)\sim (Y,t+s),\]
\[
(X,0)\sim (Z,s),
\] and
\[
(Z/Y-1,0)\sim (W,s+t).
\] It holds that $\abs{W}\leq \Theta$.

Why does the axiom make sense? Note that the conventional stationarity axiom requires that $Z/Y=1$ because the comparison of $X$ and $Y$ at the delay $s$ should not depend on the starting point for the comparison. Here we relax the stationarity axiom by allowing for $Z/Y\neq 1$. The relaxation is relative to $t+s$, so that larger deviations from $Y\neq Z$ are allowed when the time periods involved are closer in the future.

\begin{theorem}\label{thm:exp2}
    If $\succeq\in \bar{\mathcal{D}}$ satisfies the $(X,\Theta)$ stationarity axiom for some $X>0$ and $\Theta\geq 0$, then there is $\gamma\in (0,1)$ such that $d(t)=\gamma^t$.
\end{theorem}

Theorem~\ref{thm:exp2} is a case in which the relaxation of the axiom results in a representation that fits the exact, non-relaxed, axiom. By relaxing stationarity in this setting, we still obtain a utility representation by means of an exponential discount factor.

\subsection{Continuous time}\label{sec:conttime}

We now turn to continuous time, and obtain a rather different kind of result. In particular, let $M = (0,\bar X)$ and $T=\Re_+$. The existence of an upper bound on monetary rewards is technically convenient. We disregard the special zero payments, as the discounted value of zero is uninteresting.

A preference $\succeq$ is strictly monotonic if, whenever $(X,t)$ and $(X',t')$ are such that $X\geq X'$ and $t\leq t'$, with at least one strict inequality, then $(X,t)$ is strictly preferred to $(X',t')$. 

Strictly monotonic preferences satisfy two natural properties in economics. First, that money is always a desirable good. Second, that delay is bad. Meaning that receiving any given monetary quantity earlier is strictly better than receiving it later. 

For convenience, we work with $x=\log X$. Then the relevant choice space becomes $(-\infty, \bar x]\times [0,+\infty)$, where $\bar x=\log \bar X$.

We entertain the following collection of axioms, where $\ep>0$:
\begin{enumerate}
\item $\succeq$ is continuous and strictly monotonic.
\item \label{ax:npv} For each $(x,t)$ there exists $x'$ with $(x,t)\succeq (x',0)$.
\item \label{ax:deprec} 
For any $x$, there is $\bar{t}$ for which $(x,0)\succ (\bar x,t)$ for any $t\geq \bar{t}$.
\item \label{ax:statep}  $\ep$-Stationarity: \\
$(\bar x,t)\sim (x,0)$ implies that, for any $\Delta>0$, $(\bar x,t+\Delta')\sim (x,\Delta)$ for some $\Delta'$ with $\abs{\Delta-\Delta'}\leq \ep$.
\item \label{ax:lipsch} $\la$-Lipschitz: \\
$(x,t)\sim (y,0)$ implies that, for any $\Delta$, $(x,t+\Delta)\sim (y + \Delta',0)$ for some $\Delta'$ with $\abs{\Delta'}\leq \la \Delta$.
\end{enumerate}

The interpretation of these axioms is, we think, straightforward. The first axiom embodies continuity, that money is desirable, and that delays are  undesirable. The second axiom means that any delayed payment is preferred to some smaller payment in the present (no delay is so bad that its present value is $-\infty$). The third says that delays can be bad enough to ``discount'' the maximum payment to any desirable level of present consumption. The final axiom establishes a bound on marginal rates of substitution --- this is our most questionable axiom, but it should be more palatable in the present setting with the $\log X$ transformation of money.

Finally, the main axiom that we want to draw attention to here is~\ref{ax:statep}. It is a relaxed version of the familiar stationarity axiom that ensures dynamically consistent decision-making.

\begin{theorem}\label{thm:exp3}
If $\succeq$ satisfies the axioms, then it admits a continuous utility representation $u$ for which there exists a continuous and strictly monotonic function $\gamma:(-\infty,\bar x)\to \Re_+$ with
\[\sup\left\{\abs{u(x,t) - \gamma^{-1}(\gamma(x)+t)}: (x,t)\in (-\infty,\bar x)\times\Re_+ \right\} \leq\la \ep.
\]
Moreover, $u$ and $\gamma$ satisfy $\gamma(\bar x)=0$, $u(\bar x,0)=\bar x$ and $$\sup\{\gamma(x):x\in (-\infty,\bar x) \}=-\inf\{u(\bar x,t):t\in\Re_+ \}=-\infty.$$
\end{theorem}

\begin{remark}
    The function $\gamma$ is directly elicited from $\succeq$ by $(\bar x,\gamma(x))\sim (x,0)$, so it is easy to provide conditions under which it is linear: $\gamma(x) = b(\bar x - x)$. This gives an approximate representation that is ordinally equivalent to $v(X,t)=Xe^{-t/b}$, where $X$ is the original consumption of money, before the logarithmic transformation. So we obtain the standard model of exponential discounting, with discount factor $e^{-1/b}$.
\end{remark}

\begin{remark}
    When $\ep$-stationarity holds with $\ep=0$, the theorem delivers an exact representation in the form $\gamma^{-1}(t+\gamma(x))$. In this case, the $\la$-Lipschitz axiom can be dispensed with. 
\end{remark}

\section{Conclusion}
Decision theory papers are usually about a representation theorem: offering behavioral axioms that characterize the optimization of a particular kind of utility function. We suggest that an approximate version of this exercise can be interesting and potentially useful. The relaxation, or approximate satisfaction, of the axioms can guarantee the approximate optimization of a utility function that satisfies the exact version of the axiom. 

The results in the paper are relatively simple. Many of them are obtained by translating known results on the stability of functional equations to an economic model.  Given the close relationship between the stability of functional equations and our applications, we think that there are likely many other questions in decision theory that can be analyzed in a similar way, and for which approximate results such as ours are viable. One could arguably ask about approximate versions of most research in decision theory. Having a better understanding of the scope of these approximation results, and obtaining better quantitative guarantees, can be important empirically to assess the severity of any observed violations of the theory.

%%%%
\section{Proofs}\label{sec:proofs}

\subsection{Proof of Theorem~\ref{thm:EU}}
 Suppose without loss of generality that the worst element $\ul p$ can be taken to be the degenerate lottery on prize $d+1$. Suppose also that the best element $\bar p$ is the degenerate lottery on the first prize. Let $\Delta_{-}=\{x\in\Re^d_+ \mid \sum_{i=1}^d x_i \leq 1 \}$.

First, we may identify the elements of $\Delta$ and of $\Delta_{-}$ by mapping each $x\in \Delta$ onto $(x_1,\ldots,x_d)\in\Delta_{-}$. In consequence, we regard $\succeq$ as a preference on $\Delta_{-}$. Note then that the origin $0$ corresponds to the degenerate lottery on the worst prize. The degenerate lottery on the best prize is $e_1$, the first unit base vector. The rest of the unit base vectors are denoted $e_i$, $i=2,\ldots,d$.

By continuity and monotonicity, for each $x\in\Delta_{-}$ there exists a unique $\al\in[0,1]$ with $x\sim \al e_1+(1-\al) 0$.  Define $u:\Delta_{-}\rightarrow\Re$ by $u(x)=\al$.

\begin{lemma}\label{lem:approxCauchy}Let $x_1,\ldots, x_k\in\Delta_{-}$ and $\la_1,\ldots,\la_k\geq 0$ with $\sum_{i=1}^k \la_i = 1$. Then
    \[
    \abs{u(\sum_{j=1}^k \la_j x_j) - \sum_{j=1}^k \la_j u(x_j)}<(k-1)\cdot \ep
    \]
\end{lemma}

To prove the lemma, we proceed by induction. First we consider the case when $k=2$. Let $p=x_1$, $q=x_2$, and $\la_1 = \la$.  
Then by the definition of $u$, 
\begin{align*}
p & \sim u(p)\bar p + (1-u(p))0 \\
q & \sim u(q)\bar p + (1-u(q))0 \\
\la p + (1-\la)q & \sim u(\la p + (1-\la)q) \bar p + (1-u(\la p + (1-\la)q))0.
\end{align*} By $\ep$-Reduction of compound lotteries, then 
\[\abs{u(\la p + (1-\la)q) - \la u(p) - (1-\la)u(q) }<\ep\]

And thus the inequality holds for $k=2$. We now complete the proof by induction. Suppose that $k\geq 3$ and that the inequality holds for $k'=2,\ldots,k-1$. If $\la_1=1$ then there is nothing to prove, so suppose that $\la_1<1$. Then if $q=\sum_{j=2}^k\frac{\la_j}{1-\la_1} x_j$ we have by the inductive hypothesis that  \[
\abs{u(q)-\sum_{j=2}^k\frac{\la_j}{1-\la_1} u(x_j)}<(k-2)\cdot \ep.
\] Hence, 
\begin{align*}
\abs{u(\sum_{j=1}^k\la_j x_j) - \sum_{j=1}^k\la_j u(x_j)} & =     
\left| u(\la_1 x_1+(1-\la_1)q)-\la_1 u(x_1) - (1-\la_1)u(q)\right. \\
& \left.+(1-\la_1)[u(q)-\sum_{j=2}^k\frac{\la_j}{1-\la_1} u(x_j)]\right| \\
& \leq \abs{u(\la_1 x_1+(1-\la_1)q)-\la_1 u(x_1) - (1-\la_1)u(q)} \\
& +(1-\la_1)\abs{u(q)-\sum_{j=2}^k\frac{\la_j}{1-\la_1} u(x_j)} \\
&<\ep + (k-2)\ep,
\end{align*} where the final inequality follows from $\ep$--reduction of compound lotteries and the inductive hypothesis. This establishes Lemma~\ref{lem:approxCauchy}.

Now we can finish the proof. Note that we have defined $u(e_i),$ $i=1,\ldots, d$ as well as $u(0)$. For each $x\in \Delta_{-}$, we have $x=(1-\sum_{i=1}^d x_i) 0 +\sum_{i=1}^d x_i \cdot e_i$ as the unique representation of $x$ as a linear combination of $e_1,\dots,e_d$ and $0$ whose coordinates sum to one. 

Define $\ell(x) = \sum_{i=1}^d x_i u(e_i)$. Note that $\ell:\Delta_{-}\to\Re$ is affine, $u(e_i)=\ell(e_i)$ for all $i$, and $u(0)=\ell(0)$. 

Now consider $x$ which is induced by a nondegenerate lottery.  Using Lemma~\ref{lem:approxCauchy}, then, we obtain that
\begin{align*}
    \abs{u(x)-\ell(x)} & =  \abs{u(\sum_{i=1}^d  x_i \cdot e_i ) - \sum_{i=1}^d  x_i \ell(e_i)} \\
& = \abs{u(\sum_{i=1}^d  x_i\cdot e_i ) - \sum_{i=1}^d x_i u(e_i)} \\
& < (\mbox{supp}(x)-1)\ep,
\end{align*}
where $\supp (x) = |\{i \mid x_i > 0\}|+\mathbf{1}_{\{\sum_{i=1}^dx_i<1\}}$.
Note finally that we may regard $u$ and $\ell$ as functions on $\Delta$ by taking $u(\delta_{d+1})=\ell(\delta_{d+1})=0$. Then $\ell(\al p + (1-\al) q) = \al \ell(p) + (1-\al) \ell(q)$. So $\ell:\Delta\to [0,1]$ is affine.  Finally, no other affine function coincides with $u$ on the degenerate lotteries.

Conversely, suppose that $(u,\ell)$ are as in the statement of the theorem. Suppose that  
$u(p_0) = u(\al_0 \bar p + (1-\al_0)\ul p)$, 
$u(p_1) = u(\al_1 \bar p + (1-\al_1)\ul p)$ and 
\[u(\la p_0 + (1-\la) p_1)  = u(\al \bar p + (1-\al) \ul p).\]
Recall that $\ell(\bar p) = \ell(e_1)=1$ and $\ell (\ul p) = \ell(0)=0$. Then $\ell(\al_i \bar p + (1-\al_i)\ul p) = \ell (\al_i \bar p) = \al_i$.

We also have that 
\[\abs{\ell(p_i)-\al_i}=
\abs{\ell(p_i) - \ell (\al_i \bar p)}\leq 
\abs{\ell(p_i) - u(p_i)}+\abs{u(p_i)-u(\al_i \bar p)}+\abs{u(\al_i \bar p)-\ell(\al_i \bar p)}\leq 2\ep.
\] Thus,
\[\begin{split}
\abs{\ell(\la p_0 + (1-\la) p_1) - (\la \al_0 + (1-\la) \al_1)}
=
\abs{\la [\ell(p_0) - \al_0] + (1-\la) [\ell(p_1) - \al_1]} \\
< 2\ep
\end{split}\] 

Then,
\begin{align*}
\abs{\al - (\la \al_0 + (1-\la) \al_1)}  
& < \abs{\ell(\al \bar p) - \ell ((\la \al_0 + (1-\la) \al_1)\bar p )} + 2\ep\\
& = \abs{\ell(\al \bar p) - u(\al \bar p) + u(\al\bar p) - \ell ((\la \al_0 + (1-\la) \al_1)\bar p )} + 2\ep\\
& \leq  \abs{\ell(\al \bar p) - u(\al \bar p)} \\
& + \abs{u(\al \bar p + (1-\al) \ul p) - \ell ((\la \al_0 + (1-\la) \al_1)\bar p )} + 2\ep\\
& < 4\ep,
\end{align*}
where the last inequality uses that $u(\la p_0 + (1-\la) p_1)  = u(\al \bar p + (1-\la) \ul p)$ and $\norm{u-\ell}_{\infty}<\ep$.

\subsection{Proof of Theorem~\ref{thm:EU2}}

First, let $\bar p, \ul p\in\Delta$ be such that $\bar p\succeq p\succeq \ul p$ for all $p\in \Delta$. By continuity and monotonicity, for any $p\in\Delta$ there is $\al\in [0,1]$ with $p\sim \al \bar p+(1-\al)\ul p$. Define $u(p) = \al$.

    \begin{lemma}\label{lem:approx0} If $p_k\sim q_k$, for $k=1,\ldots,K$, then for any $\la_k\geq 0$ with $\sum_k \la_k = 1$ there exists $\mu_k\geq 0$ with $\sum_k \mu_k = 1$ and $\abs{\la_k-\mu_k} < K\ep$ such that
    \[\sum_{k=1}^K \la_k p_k  \sim  \sum_{k=1}^K \mu_k q_k \]
        \end{lemma}

\begin{proof} We may assume without loss that $\la_k>0$ for all $k$.

    Since $p_1\sim q_1$, there is $\la'_1$ with $|\la_1-\la'_1|<\ep$ such that
\[\la_1 p_1 + (1-\la_1) \sum_{k=2}^K \frac{\la_k}{\sum_{j=2}^K \la_j} p_k
\sim \la'_1 q_1 + (1-\la'_1) \sum_{k=2}^K \frac{\la_k}{\sum_{j=2}^K \la_j} p_k
= \la'_1 q_1 + \sum_{k=2}^K \la'_k p_k,
\] where \[
            \la'_k \coloneqq (1-\la'_1)\frac{\la_k}{\sum_{j=2}^K\la_j}.
    \] Observe that
    \[| \la_k-\la'_k | = |\la_k \left( 1-\frac{1-\la'_1}{1-\la_1}\right) |
    = \frac{\la_k}{1-\la_1} | \la'_1-\la_1 | < \ep,\]
    as $\la_k\leq 1-\la_1$.

    By repeating this argument $K$ times, we obtain that
    \[\sum_{k=1}^K \la_k p_k \sim \sum_{k=1}^K \mu_k q_k,
    \] with $|\la_k-\mu_k|<K\ep$.
\end{proof}

    \begin{lemma}\label{lem:approx} Let $p_k\in \Delta$, $k=1,\ldots,K$. Suppose that $\la_k\geq 0$ with $\sum_k \la_k = 1$. Then,
    \[| u(\sum_{k=1}^K \la_k p_k ) - \sum_{k=1}^K \la_k u(p_k)  |<K^2\ep \]
        \end{lemma}
    \begin{proof} Observe that $$ p_k\sim u(p_k)\bar p + (1-u(p_k))\ul p$$ for each $k$. Hence, by Lemma~\ref{lem:approx0}, we have 
\[\begin{split}
\sum_{k=1}^K \la_k p_k \sim \sum_{k=1}^K \mu_{k} [u(p_k)\bar p + (1-u(p_k))\ul p] = \sum_{k=1}^K \mu_k u(p_k)\bar p \\ +   \sum_{k=1}^K \mu_k(1-u(p_k)) \ul p,
        \end{split}
\] with $|\la_k-\mu_k|<K\ep$.

    Thus we have, by definition of $u$, that\[u(\sum_{k=1}^K \la_k p_k ) = \sum_{k=1}^K \mu_k u(p_k).\]

    Finally,
        \begin{align*}
            |u(\sum_{k=1}^K \la_k p_k ) - \sum_{k=1}^K \la_k u(p_k )|
            & \leq
        |u(\sum_{k=1}^K \la_k p_k ) - \sum_{k=1}^K \mu_k u(p_k)| \\
            & +| \sum_{k=1}^K \la_k u(p_k )- \sum_{k=1}^K \mu_k u(p_k) | \\
              & \leq \sum_{k=1}^K | \la_k- \mu_k| u(p_k) \leq K^2 \ep,
        \end{align*} as $u(p_k)\in [0,1]$ for each $k$.
    \end{proof}

Define $\ell(p)=\sum_{i=1}^{d+1} p_i u(e_i)$. By Lemma~\ref{lem:approx}, $\abs{u(p)-\ell(p)}\leq (d+1)^2\ep$.

\subsection{Proof of Theorem~\ref{thm:AA}}

Define $u: X \to \Re$ by $u(x) = c$, where $c \in \Re_+$ is the unique scalar such that $x \sim c\mathbf{1}$. (Recall the convention $c \equiv c\mathbf{1}$.) By monotonicity and continuity, $u$ is well defined. If $u(x)=c_x$ and $u(y)=c_y$, then $\phi$-approximate independence implies that 
\[
\frac{1}{2}c_x +\frac{1}{2}c_y + \phi(x,y)
\succeq \frac{1}{2} x + \frac{1}{2} y
\succeq \frac{1}{2}c_x +\frac{1}{2}c_y - \phi(x,y).
\] Consequently, by definition of $u$,
\[\abs{u\left(\frac{x+y}{2}\right) - \frac{1}{2}u(x) - \frac{1}{2}u(y)}\leq \phi(x,y).
\]

  In particular, for $y=0$, we have $u(y)=0$, so we obtain 
\begin{equation}\label{eq:basecase}
    \abs{\frac{1}{2}u(2x) - u(x)}\leq \phi(2x,0).
\end{equation}
Define $\da(x) =\phi(x,0)$.

Repeating essentially the same arguments as in \cite{jung1998hyers}, we can prove by induction that for any $n$,
\begin{equation}\label{eq:ineqn}
\abs{2^{-n}u(2^nx) - u(x)}\leq \sum_{i=1}^n 2^{-(i-1)}\da(2^ix).
\end{equation}
For completeness, we include the proof. The base case is Equation~\eqref{eq:basecase}. Suppose that Equation~\eqref{eq:ineqn} holds for $n$. Note that~\eqref{eq:basecase} implies that
  \[
      \abs{\frac{1}{2}u(2^{n+1}x) - u(2^nx)}  \leq \da(2^{n+1} x),
  \] and hence that $\abs{2^{-(n+1)} u(2^{n+1}x) - 2^{-n}u(2^nx)}  \leq 2^{-n}\da(2^{n+1} x)$.

  Now,
\begin{align*}
\abs{2^{-(n+1)}u(2^{n+1}x) - u(x)}
& \leq \abs{2^{-(n+1)}u(2^{n+1}x) -  2^{-n}u(2^n x)} + \abs{2^{-n}u(2^nx) - u(x)} \\
 &  \leq 2^{-n}\da(2^{n+1} x) + \sum_{i=1}^{n} 2^{-(i-1)}\da(2^{i} x),
\end{align*} which proves the induction step.

As a consequence, when $m> n$ we have
\[\begin{split}
\abs{2^{-m}u(2^{m}x) - 2^{-n}u(2^{n}x)}=2^{-n}\abs{2^{-(m-n)}u(2^{m-n} 2^n x) - u(2^{n}x)} \\ \leq 2^{-n} \sum_{i=1}^{m-n} 2^{-(i-1)}\da(2^i 2^n x).
\end{split}\] We have $2^{-n} \sum_{i=0}^{m-n-1} 2^{-i}\da(2^{i+1} 2^n x)\leq 2^{-n}\Theta\to 0$, and hence the sequence $\{2^{-n}u(2^n x) \}$ is Cauchy. Therefore, it is convergent.

Define $v:X\to\Re$ by $v(x)= \lim_{n\to \infty} 2^{-n}u(2^n x)$.

Then,
  \[
\abs{v(x)-u(x)} = \lim_{n\to\infty}
\abs{2^{-n}u(2^nx) - u(x)}\leq \limsup_{n\to\infty}
\sum_{i=0}^{n-1} 2^{-i}\da(2^{i+1}x)\leq \Theta.
\]

Observe that $v$ is monotone increasing because $u$ is monotone, and $v(0)=0$ because $u(0)=0$.

  Note that, since $\sum_{i=0}^{n} 2^{-i}\phi(2^{i}x,2^i y)$ is bounded, we must have that $2^{-n}\phi(2^{n}x,2^{n}y)\to 0$ as $n\to\infty$. Then, by definition of $v$:
\begin{align*}
\abs{v\left(\frac{x+y}{2}\right) - \frac{1}{2}v(x) - \frac{1}{2}v(y)}
& =
\lim_{n\to\infty} 2^{-n}
\abs{u\left(\frac{2^{n} x+2^{n}y}{2}\right) - \frac{1}{2}u(2^{n}x) - \frac{1}{2}u(2^{n}y)} \\
& \leq
\lim_{n \to \infty} 2^{-n} \phi(2^n x, 2^n y) = 0.
\end{align*}

So $v$ satisfies Jensen's equation, that is $v(\frac{x+y}{2})=\frac{v(x)+v(y)}{2}$. But then we know that $2v(\frac{1}{2}x)= v(x)$, as $v(0)=0$, and therefore that
\[
    v(x+y) = 2v(\frac{1}{2}(x+y)) = v(x)+v(y).
\]
Hence, $v$ is additive on $X$ (satisfies the Cauchy equation).  We have also established previously that it is monotone.  Standard arguments then imply that $v$ must be linear; we proceed to spell this out.

Additivity of $v$ implies that $v(kx)=kv(x)$ for all positive integers $k$. Thus $qv(x) = v(qx)$ for all rational $q>0$ by the usual argument for the Cauchy equation.\footnote{For any two positive integers $k,\ell$,
$\ell v(x) = v(\ell x) = v(k \frac{\ell x}{k})= k v(\frac{\ell x}{k})$.}

Consider $r\in\Re_+$ and let $\ul q_k$, $\bar q_k$ be sequences of rational numbers that converge to $r$ and for which $\ul q_k\leq r\leq \bar q_k$. Then, by the monotonicity of $v$, we have
\[
\ul q_k v(x) = v(\ul q_k x) \leq v(rx) \leq v(\bar q_k x)=\bar q_k v(x).
\] Then $\ul q_k\to r$ and $\bar q_k\to r$ imply that $rv(x)=v(rx)$.

In consequence, for any $x\in\Re^d_+$, $v(x)= \sum_{i=1}^d x_i v(e_i)$. We may then extend $v$ to $\Re^d$ by defining $v(x)=\sum_{i=1}^d x_i v(e_i)$ for all $x\in \Re^d$. So $v$ is linear.

%%%%%%%%%%%%%%%%%%%%%%%%
To prove the second statement in the theorem, suppose now that $\succeq$ is homothetic. 
Note that $u$, as defined, is homogeneous of degree one. By homogeneity of $u$, then, $u(x)=2^{-n} u(2^n x)$. So $v(x)=u(x)$.

\subsection{Proof of Corollary~\ref{cor:smoothamb}}

First we deal with the case when  $u(x) = \int_{\Delta} \sqrt{1+(\sum p_i x_i)^2} \, d\mu(p)$. Following the steps in the proof of Theorem~\ref{thm:AA}, we define the candidate linear approximation $v$ as the limit of the scaled utility:$$v(x) = \lim_{n \to \infty} 2^{-n} u(2^n x).$$Substituting the explicit form of $u$:
\[
    v(x) = \lim_{n \to \infty} 2^{-n} \int_{\Delta} \sqrt{1 + (2^n \sum p_i x_i)^2} \, d\mu(p).
\]
Since $f(z) = \sqrt{1+z^2}$ is asymptotically linear (specifically, $\lim_{\alpha \to \infty} \frac{\sqrt{1+(\alpha z)^2}}{\alpha} = z$ for $z>0$), we can pass the limit inside the integral (justified by the Dominated Convergence Theorem, as the difference is uniformly bounded):
\[
    v(x) = \int_{\Delta} \left( \sum p_i x_i \right) d\mu(p) = \sum_{i=1}^d x_i \left( \int_{\Delta} p_i \, d\mu(p) \right) = \sum_{i=1}^d x_i \bar{p}_i.
\]
Thus, $v$ is well-defined and linear.

To bound the error $|u(x) - v(x)|$, we use the triangle inequality on the integral representation derived from the limit definition:
\[
    |u(x) - v(x)| = \left| \int_{\Delta} \left( \sqrt{1+(\sum p_i x_i)^2} - \sum p_i x_i \right) d\mu(p) \right|.
\]
Consider the scalar function $h(z) = \sqrt{1+z^2} - z$ for $z \geq 0$. Observe that $h(0)=1$, $h(z) > 0$, and $h'(z) = \frac{z}{\sqrt{1+z^2}} - 1 < 0$. Thus, $h(z)$ is monotonically decreasing with a maximum value of $1$ at $z=0$.

Applying this bound pointwise to the integrand (where $z = \sum p_i x_i \geq 0$):
\[
    \left| \sqrt{1+(\sum p_i x_i)^2} - \sum p_i x_i \right| \leq 1.
\]
Therefore,
\[
    |u(x) - v(x)| \leq \int_{\Delta} 1 \, d\mu(p) = 1.
\]
This confirms that the global approximation error is bounded by $1$.

Now we consider the case of $f(z)=z-e^{-z}$.

We derive the linear limit $v$ using the definition from Theorem~\ref{thm:AA}:
    \[
        v(x) = \lim_{n \to \infty} 2^{-n} u(2^n x) = \lim_{n \to \infty} 2^{-n} \int_{\Delta} \left( 2^n \sum p_i x_i - e^{-2^n \sum p_i x_i} \right) d\mu(p).
    \]
    Splitting the integral by linearity:
    \[
        v(x) = \lim_{n \to \infty} \left[ \int_{\Delta} \left(\sum p_i x_i\right) d\mu(p) - 2^{-n} \int_{\Delta} e^{-2^n \sum p_i x_i} d\mu(p) \right].
    \]
    The first term is constant with respect to $n$ and equals $\sum x_i \bar{p}_i$. The second term is bounded by $2^{-n}$ (since the exponential is $\leq 1$ for non-negative arguments) and thus converges to 0 as $n \to \infty$. Hence, $v(x) = \sum x_i \bar{p}_i$.
    
    To bound the error, consider the difference:
    \[
        |u(x) - v(x)| = \left| \int_{\Delta} \left( \sum p_i x_i - e^{-\sum p_i x_i} \right) d\mu(p) - \int_{\Delta} \left(\sum p_i x_i\right) d\mu(p) \right|.
    \]
    Simplifying, we obtain:
    \[
        |u(x) - v(x)| = \left| -\int_{\Delta} e^{-\sum p_i x_i} d\mu(p) \right| = \int_{\Delta} e^{-\sum p_i x_i} d\mu(p).
    \]
    Since $x \in \mathbb{R}^d_+$ and $p \in \Delta$, we have $\sum p_i x_i \geq 0$, which implies $0 < e^{-\sum p_i x_i} \leq 1$. Therefore:
    \[
        |u(x) - v(x)| \leq \int_{\Delta} 1 \, d\mu(p) = 1.
    \]
    Thus, the global approximation error is bounded by 1.

\subsection{Proof of Theorem~\ref{thm:approxhomog}}

The proof of Theorem~\ref{thm:approxhomog} is essentially based on \cite{czerwik1992stability}, but following the ideas from Jung that were used in the proof of the previous result. \cite{hyers2012stability} contains some analogous results. The result of Czerwik does not apply directly because his hypotheses are not satisfied. 

The utility function $u$ is defined just as in the previous proof. Then, $x\sim u(x)$ implies that 
$\la u(x) + \bar\phi(u(x)\one,x,\la)\one \succeq \la x$ and 
$\la x \succeq \la u(x) - \ul \phi(x,u(x)\one,\la)\one$. Hence, 
\[ 
\abs{u(\la x) - \la u(x)}\leq \max\{\bar\phi(x,u(x)\one,\la), \ul \phi(x,u(x)\one,\la)\}.
\]

Note that since $\la^n\max\{x_i\}\one\geq \la^n x\geq \la^n \min\{x_i\}\one$, monotonicity of $\succeq$ implies that $\la^n\max\{x_i\}\one\geq u(\la^n x)\geq \la^n \min\{x_i \}$.  This implies that
\begin{align*}
    \bar\phi(\la^n\min\{x_i\}\one,\la^n x,\al) & \geq \bar\phi(u(\la^n x),\la^n x,\al) \\
    \text{ and }
\ul\phi(\la^n x,u(\la^n x),\al)& \leq \ul\phi(\la^n x,\la^n \max\{x_i\}\one,\al).
\end{align*}

Let $\phi(x,\la) \coloneqq \max\{\bar\phi(x,u(x)\one,\la), \ul \phi(x,u(x)\one,\la)\}$.
\[\phi(\la^n x,\al)\leq  \max\{\bar\phi(\la^n\min\{x_i\}\one,\la^n x,\al),\ul\phi(\la^n x,\la^n \max\{x_i\}\one,\al).
\]
Hence, given the hypotheses on $\bar\phi$ and $\ul\phi$, we obtain that 
\begin{equation}\label{eq:homogvanish}
 \la^{-n}\phi(\la^n x,\al)\to 0.
\end{equation}

We claim that 
\begin{equation}\label{eq:homoginduc}
  \abs{\la^{-n}u(\la^n x) - u(x)}\leq \sum_{j=0}^{n-1}\la^{-(j+1)}\phi(\la^j x,\la).  
\end{equation}
The base case ($n=1$) is clear from the inequality:
\[
\abs{\la^{-1}u(\la x) - u(x)}\leq \frac{1}{\la}\phi(x,\la)
\]
So suppose that~\eqref{eq:homoginduc} is true for $n$. Then the base case with $\la^n x$ in place of $x$ yields
\[
\abs{\la^{-1}u(\la \la^n x) - u(\la^n x)}\leq \frac{1}{\la}\phi(\la^n x,\la).
\]
So we have 
\begin{align*}
\abs{\la^{-(n+1)}u(\la^{n+1} x) - u(x)} & \leq 
\la^{-n}\abs{\la^{-1}u(\la \la^n x) - u(\la^n x)}  + \abs{\la{-n}u(\la^n x) - u(x)} \\
& \leq \la^{n+1}\phi(\la^n x,\la) + \sum_{i=0}^{n-1}\la^{-(j+1)}\phi(\la^j x,\la),    \end{align*}
which proves the inequality~\eqref{eq:homoginduc} by induction.

Next, we show that the sequence $\{\la^{-n} u(\la^n x)\}$ has the Cauchy property. Indeed, assuming that $n> m$,
\begin{align*}
\abs{\la^{-n}u(\la^n x) - \la^{-m} u(\la^{m} x)}
& = \la^{-m}\abs{\la^{-(n-m)}u(\la^{n-m}\la^m x) - u(\la^{m} x)} \\
&\leq \la^{-m} \sum_{j=0}^{n-m-1}\la^{-(j+1)}\phi(\la^{j+m} x,\la)= \sum_{j=m}^{n-m-1}\la^{-(j+1)}\phi(\la^{j} x,\la) \\ 
& \leq 
\sum_{j=m}^{n-m-1}\la^{-(j+1)} \bar\phi(\la^{j} x,u(\la^{j} x)\one,\la) \\
& +
\sum_{j=m}^{n-m-1}\la^{-(j+1)} \ul\phi(\la^{j} x,u(\la^{j} x)\one,\la).
\end{align*}

By the Cauchy property, the sequence is convergent. So we may define $v(x) \coloneqq \lim_{n\to \infty} \la^{-n}u(\la^n x)$. Then we have, for any $\al> 0$, that 
\[
\abs{v(\al x) - \al v(x)}
= \lim_{n\to\infty}
\abs{\la^{-n}u(\al \la^n x) - \la^{-n}\al u(\la^n x)}
\leq \lim_{n\to\infty}
\la^{-n}\phi(\la^n x,\al) = 0, 
\] using Equation~\eqref{eq:homogvanish}. So $v$ is homogeneous.

Finally, 
\[
\abs{v(x)-u(x)} = \lim_{n\to\infty} \abs{\la^{-n}u(\la^n x) - u(x)} \leq \sum_{j=0}^{\infty}\la^{-(j+1)}\phi(\la^j x,\la) \leq 2\Theta.  
\]

\subsection{Proof of Theorem~\ref{thm:quasiquasicon}}

Let $U(c)=\{x \mid x\succeq c \}$ be the upper contour set at $c\one$, and $\bar U(c)$ the convex hull of $U(c)$. Define 
$u(x)=\sup\{c \mid x\in U(c) \}$ and
$v(x)=\sup\{c \mid x\in \bar U(c) \}$.
By monotonicity and continuity, $u,v:\Re^d_+\to\Re$ are well defined and $x\sim u(c)\one$. So $u$ represents $\succeq$.

The function $v$ is quasi-concave because $v(x)\geq v(z)$ and $v(y)\geq v(z)$ imply that $x,y\in \bar U(v(z))$. So $\la x + (1-\la) y\in \bar U(v(z))$, and hence $v(\la x + (1-\la) y)\geq v(z)$.

Note that $u(x)\leq v(x)$, as $U(c)\subseteq \bar U(c)$ for all $c$. If $u(x)<v(x)$ then $x\in \bar U(v(x))\setminus U(v(x))$. 

By a version of Carathéodory's theorem (see \cite{hanner51}), there are at most $d$ points $x_1,\ldots, x_K$ in $U(v(x))$ with $x=\sum_{i=1}^K\la_i x_i$, $\la_i>0$, and $\sum_{i=1}^K\la_i =1$. 

By the $\ep$-uncertainty aversion axiom, if $\frac{1}{\sum_{i=1}^k \la_i} \sum_{i=1}^k \la_ix_i \succeq c$ and $x_{k+1}\succeq c$, then 
\[
\left( \frac{\sum_{i=1}^k \la_i }{\sum_{i=1}^{k+1} \la_i}\right)
\frac{1}{\sum_{i=1}^k \la_i} \sum_{i=1}^k \la_ix_i + \frac{\la_{k+1}}{\sum_{i=1}^{k+1} \la_i}x_{k+1}\succeq c-\ep.
\]
Since $x_1\succeq v(x)$, and $x_{k+1}\succeq v(x) \succeq v(x)-k \ep$, for $k=1,\ldots, K-1$, we conclude that $\sum_{i=1}^k \la_i x_i \succeq v(x)-k\ep\leq v(x)-d\ep$ (as $K\leq d$). Hence $x\in U(v(x)-d\ep)$, and thus $u(x)\geq v(x)-d\ep$.

\subsection{Proof of Theorem~\ref{thm:exp}}

The ideas are very similar to what we have already used in other results, so we only offer a sketch of the proof. It is perhaps worth mentioning that the proof requires slightly weaker conditions than were stated in the theorem. In particular, that 
\[
\sum_{i=0}^{\infty} 2^{-(i+1)} \abs{\log \da_x(2^i t, 2^i t)}
\] is summable provides the approximation guarantee, but we also need that \[
\lim_{n\to\infty} 2^{-n}\da_x(2^n t, 2^n s) = 0
\] for the Cauchy property below. 

We proceed to briefly sketch the proof.

Suppose that ${\succeq}\in \mathcal{D}$ and let the associated discount factor be $d$, normalized so that $d(1)=1$. Then we have 
\[
xd(t) = yd(s+t) \text{ and } x = \da_x(s,t) y f(s).
\]
Hence, we have 
\[ \frac{d(s+t)}{d(s)d(t)} = \da_x(s,t).
\] Letting $g=\log d$, we obtain
\[
\abs{g(s+t) - g(s) - g(t)} = \psi_x(s,t) \coloneqq \abs{\log \da_x(s,t)}.
\]

Setting $s=t$ we see that:
\[
\abs{\frac{1}{2}g(2t) - g(t)} =  \frac{1}{2}\psi_x(t),
\] where we write $\psi_x(t)$ for $\psi_x(t,t)$ as a notational shortcut.

Now, just as in other results, we can show that $\abs{2^{-n}g(2^n t) - g(t)}$ is bounded above by $\sum_{i=0}^{n-1} 2^{-(i+1)}\psi_x(2^i t)$, and use this to prove the Cauchy property of the sequence $2^{-n}g(2^n t)$. So we may define $$h(t)= \lim_{n\to \infty} 2^{-n} g(2^n t),$$ and show that 
\[
\abs{g(t)- h(t)}\leq \sum_{i=0}^{\infty} 2^{-(i+1)}\psi_x(2^i t).
\] Again, familiar arguments establish that $h$ is additive. 

The only solution to the Cauchy equation over $T$ is $h(t)=c t$. Hence we consider $e^{h(t)} = \g^t$, where $\g\coloneqq e^{c}$. Note that 
$$\g = e^{h(1)} = \lim_{n\to\infty}  2^{-n}\log d(2^n).$$
Finally,
\[
\Theta\geq \abs{g(t)- h(t)} = \abs{\log d(t) - \log \g^t} = \abs{\log \frac{d(t)}{\g^t}}.
\]    

\subsection{Proof of Theorem~\ref{thm:exp2}}

From the axiom we see that $Xd(t)=Yd(t+s)$ and $X = Zd(s)$, so 
\[
\frac{d(t+s)}{d(t)d(s)} - 1 = \frac{Z}{Y} -1 = W d(t+s),
\] so 
\[
\abs{f(t+s) - f(s)f(t)} = \abs{W}\leq \Theta,
\] where $f=1/d$. Since $d(t)\to 0$, $f$ is unbounded. By Theorem~2 of \cite{baker1979}, it follows that $f$, and therefore, $d$, is exponential. Theorem~2 in \cite{baker1979} is stated for functions that are defined over the integers, but it is clear from the proof that the result applies in our case. 

It is perhaps worth pointing out that, given $d\in\bar{\mathcal{D}}$ satisfying the $(X,\Theta)$ axiom, the number $\gamma$ is defined as follows: Let $\tau$ be such that $d(\tau)\leq \min\{\frac{1}{4},\frac{1}{4\Theta} \}$. Then $\gamma = \sqrt[\tau]{d(\tau)}$.

\subsection{Proof of Theorem~\ref{thm:exp3}}

\begin{proof}
For any $(x,t)$ with $t>0$, $(x,0)\succ (x,t)$. By the second axiom, $(x,t)\succeq (x',0)$ for some $x'$. By continuity and monotoniciy, then, there exists a unique $\hat x\in (-\infty,\bar x)$ with $(\hat x,0)\sim (x,t)$. Define a function $u:X\times\Re_+\rightarrow\Re$ by $(u(x,t),0)\sim(x,t)$. The function $u$ is a continuous utility representation of $\succeq$.

Axiom~\ref{ax:statep} implies that, for any $s>0$, there is $\eta$ with $\abs{\eta}\leq \ep$ such that
\[
(u(\bar x,t),s)\sim (\bar x,t+s+\eta).
\]

But since $(\bar x,t+s+\eta)\sim (u(\bar x,t),s)\sim (u(u(\bar x,t),s),0)$, we see by Axiom~\ref{ax:lipsch} that
\[
(\bar x,t+s)\sim (u(u(\bar x,t),s)+\eta',0),
\] with $\abs{\eta'}\leq \la \abs{\eta}$. Thus, 
\[
u(\bar x,t+s) = u(u(\bar x,t),s)+\eta'
\] with $\abs{\eta'}\leq \la \abs{\eta}\leq \la \ep$. (Observe that when $\ep=0$ we reach this conclusion with $\eta'=0$, without the use of Axiom~\ref{ax:lipsch}.)

By Axiom~\ref{ax:deprec}, for any $x\in (-\infty,\bar x]$ there is $t'$ with $(x,0)\succ (\bar x,t')$. By the continuity of $\succeq$, there exists $t$ with $(x,0)\sim(\bar x,t)$. As a consequence, $t\mapsto u(\bar x,t)$ is surjective. Given that $u$ strictly decreases monotonically in $t$, it is, in fact, a bijection from $\Re_+$ to $(-\infty,\bar x ]$. 

The rest of the proof is borrowed from \cite{mach2006stability}. Define $g(t) = u(\bar x,t)$ and $h(x,t) = g(t+ g^{-1}(x))$. Then $u(\bar x,g^{-1}(x))=x$. So,
\begin{align*}
    \abs{u(x,t) - h(x,t)} & = \abs{u(u(\bar x,g^{-1}(x)),t) - g(t+ g^{-1}(x))} \\
    & = \abs{u(u(\bar x,g^{-1}(x)),t) - u(\bar x,t+ g^{-1}(x))}\leq\la\ep.
\end{align*} Letting $\gamma=g^{-1}$ finishes the proof. Note that $u(\bar x,0)=\bar x$ and $\inf \{u(\bar x,t) \mid t\in\Re_+ \}=-\infty$ imply that $\gamma(\bar x)=0$ and $\sup \{\gamma(x):x\in (-\infty,\bar x] \}=\infty$.
\end{proof}

\bibliographystyle{ecta}
\bibliography{approxCauchy}

\begin{thebibliography}{35}
\newcommand{\enquote}[1]{``#1''}
\expandafter\ifx\csname natexlab\endcsname\relax\def\natexlab#1{#1}\fi

\bibitem[\protect\citeauthoryear{Afriat}{Afriat}{1973}]{afriat1973system}
\textsc{Afriat, S.} (1973): \enquote{{On a system of inequalities in demand
  analysis: an extension of the classical method},} \emph{International
  Economic Review}, 14, 460--472.

\bibitem[\protect\citeauthoryear{Allais}{Allais}{1953}]{allais1953comportement}
\textsc{Allais, M.} (1953): \enquote{Le comportement de l'homme rationnel
  devant le risque: critique des postulats et axiomes de l'{\'e}cole
  am{\'e}ricaine,} \emph{Econometrica}, 503--546.

\bibitem[\protect\citeauthoryear{Anderson}{Anderson}{1986}]{anderson1986almost}
\textsc{Anderson, R.~M.} (1986): \enquote{Almost implies near,}
  \emph{Transactions of the American Mathematical Society}, 296, 229--237.

\bibitem[\protect\citeauthoryear{Baker, Lawrence, and Zorzitto}{Baker
  et~al.}{1979}]{baker1979}
\textsc{Baker, J., J.~Lawrence, and F.~Zorzitto} (1979): \enquote{The Stability
  of the Equation $f(x + y) = f(x)f(y)$,} \emph{Proceedings of the American
  Mathematical Society}, 74, 242--246.

\bibitem[\protect\citeauthoryear{Blavatskyy, Panchenko, and Ortmann}{Blavatskyy
  et~al.}{2023}]{blavatskyy2023common}
\textsc{Blavatskyy, P., V.~Panchenko, and A.~Ortmann} (2023): \enquote{How
  common is the common-ratio effect?} \emph{Experimental Economics}, 26,
  253--272.

\bibitem[\protect\citeauthoryear{Camerer}{Camerer}{1995}]{camerer95}
\textsc{Camerer, C.} (1995): \enquote{Individual Decision Making,} in \emph{The
  handbook of experimental economics}, ed. by A.~E. Roth and J.~H. Kagel,
  Princeton university press Princeton, vol.~1.

\bibitem[\protect\citeauthoryear{Czerwik}{Czerwik}{1992}]{czerwik1992stability}
\textsc{Czerwik, S.} (1992): \enquote{On the stability of the homogeneous
  mapping,} \emph{CR Math. Rep. Acad. Sci. Canada}, 14, 268--272.

\bibitem[\protect\citeauthoryear{{de Clippel} and Rozen}{{de Clippel} and
  Rozen}{2023}]{declippel2023relaxed}
\textsc{{de Clippel}, G. and K.~Rozen} (2023): \enquote{Relaxed Optimization:
  How Close Is a Consumer to Satisfying First-Order Conditions?} \emph{Review
  of Economics and Statistics}, 105, 883--898.

\bibitem[\protect\citeauthoryear{DeJarnette, Dillenberger, Gottlieb, and
  Ortoleva}{DeJarnette et~al.}{2020}]{dejarnette2020time}
\textsc{DeJarnette, P., D.~Dillenberger, D.~Gottlieb, and P.~Ortoleva} (2020):
  \enquote{Time lotteries and stochastic impatience,} \emph{Econometrica}, 88,
  619--656.

\bibitem[\protect\citeauthoryear{Echenique, Imai, and Saito}{Echenique
  et~al.}{2023}]{echenique2023approximate}
\textsc{Echenique, F., T.~Imai, and K.~Saito} (2023): \enquote{Approximate
  expected utility rationalization,} \emph{Journal of the European Economic
  Association}, 21, 1821--1864.

\bibitem[\protect\citeauthoryear{Fishburn}{Fishburn}{1989}]{fishburn89}
\textsc{Fishburn, P.~C.} (1989): \enquote{Retrospective on the Utility Theory
  of von {N}eumann and {M}orgenstern,} \emph{Journal of Risk and Uncertainty},
  2, 127--157.

\bibitem[\protect\citeauthoryear{Fishburn and Rubinstein}{Fishburn and
  Rubinstein}{1982}]{fishrubin89}
\textsc{Fishburn, P.~C. and A.~Rubinstein} (1982): \enquote{Time Preference,}
  \emph{International Economic Review}, 23, 677--694.

\bibitem[\protect\citeauthoryear{Frederick, Loewenstein, and
  O'Donoghue}{Frederick et~al.}{2002}]{shane2002}
\textsc{Frederick, S., G.~Loewenstein, and T.~O'Donoghue} (2002): \enquote{Time
  Discounting and Time Preference: A Critical Review,} \emph{Journal of
  Economic Literature}, 40, 351--401.

\bibitem[\protect\citeauthoryear{Gilboa and Schmeidler}{Gilboa and
  Schmeidler}{1989}]{gilboa1989maxmin}
\textsc{Gilboa, I. and D.~Schmeidler} (1989): \enquote{Maxmin expected utility
  with non-unique prior,} \emph{Journal of mathematical economics}, 18,
  141--153.

\bibitem[\protect\citeauthoryear{Hanner and Rådström}{Hanner and
  Rådström}{1951}]{hanner51}
\textsc{Hanner, O. and H.~Rådström} (1951): \enquote{A Generalization of a
  Theorem of Fenchel,} \emph{Proceedings of the American Mathematical Society},
  2, 589--593.

\bibitem[\protect\citeauthoryear{Houtman and Maks}{Houtman and
  Maks}{1985}]{houtman1985determining}
\textsc{Houtman, M. and J.~Maks} (1985): \enquote{Determining all maximal data
  subsets consistent with revealed preference,} \emph{Kwantitatieve methoden},
  19, 89--104.

\bibitem[\protect\citeauthoryear{Hyers, Isac, and Rassias}{Hyers
  et~al.}{2012}]{hyers2012stability}
\textsc{Hyers, D.~H., G.~Isac, and T.~Rassias} (2012): \emph{Stability of
  functional equations in several variables}, vol.~34, Springer Science \&
  Business Media.

\bibitem[\protect\citeauthoryear{Jung}{Jung}{1998}]{jung1998hyers}
\textsc{Jung, S.-M.} (1998): \enquote{{H}yers-{U}lam-{R}assias stability of
  {J}ensen’s equation and its application,} \emph{Proceedings of the American
  Mathematical Society}, 126, 3137--3143.

\bibitem[\protect\citeauthoryear{Kanhneman and Tversky}{Kanhneman and
  Tversky}{1979}]{kanhneman1979prospect}
\textsc{Kanhneman, D. and A.~Tversky} (1979): \enquote{Prospect theory: An
  analysis of decision under risk,} \emph{Econometrica}, 47, 263--91.

\bibitem[\protect\citeauthoryear{Klibanoff, Marinacci, and Mukerji}{Klibanoff
  et~al.}{2005}]{klibanoff2005}
\textsc{Klibanoff, P., M.~Marinacci, and S.~Mukerji} (2005): \enquote{A Smooth
  Model of Decision Making under Ambiguity,} \emph{Econometrica}, 73,
  1849--1892.

\bibitem[\protect\citeauthoryear{Koopmans}{Koopmans}{1960}]{koopmans1960stationary}
\textsc{Koopmans, T.~C.} (1960): \enquote{Stationary ordinal utility and
  impatience,} \emph{Econometrica}, 287--309.

\bibitem[\protect\citeauthoryear{Lancaster}{Lancaster}{1963}]{lancaster63}
\textsc{Lancaster, K.} (1963): \enquote{An Axiomatic Theory of Consumer Time
  Preference,} \emph{International Economic Review}, 4, 221--231.

\bibitem[\protect\citeauthoryear{Loomes}{Loomes}{2005}]{loomes2005modelling}
\textsc{Loomes, G.} (2005): \enquote{Modelling the stochastic component of
  behaviour in experiments: Some issues for the interpretation of data,}
  \emph{Experimental Economics}, 8, 301--323.

\bibitem[\protect\citeauthoryear{Mach and Moszner}{Mach and
  Moszner}{2006}]{mach2006stability}
\textsc{Mach, A. and Z.~Moszner} (2006): \enquote{On stability of the
  translation equation in some classes of functions,} \emph{Aequationes
  Mathematicae}, 72, 191--197.

\bibitem[\protect\citeauthoryear{Machina}{Machina}{1987}]{machina87}
\textsc{Machina, M.~J.} (1987): \enquote{Choice Under Uncertainty: Problems
  Solved and Unsolved,} \emph{The Journal of Economic Perspectives}, 1,
  121--154.

\bibitem[\protect\citeauthoryear{McGranaghan, Nielsen, O’Donoghue,
  Somerville, and Sprenger}{McGranaghan
  et~al.}{2024}]{mcgranaghan2024distinguishing}
\textsc{McGranaghan, C., K.~Nielsen, T.~O’Donoghue, J.~Somerville, and C.~D.
  Sprenger} (2024): \enquote{Distinguishing common ratio preferences from
  common ratio effects using paired valuation tasks,} \emph{American Economic
  Review}, 114, 307--347.

\bibitem[\protect\citeauthoryear{McKelvey and Palfrey}{McKelvey and
  Palfrey}{1995}]{mckelvey1995quantal}
\textsc{McKelvey, R.~D. and T.~R. Palfrey} (1995): \enquote{Quantal response
  equilibria for normal form games,} \emph{Games and economic behavior}, 10,
  6--38.

\bibitem[\protect\citeauthoryear{Ok and Masatlioglu}{Ok and
  Masatlioglu}{2007}]{OK2007214}
\textsc{Ok, E.~A. and Y.~Masatlioglu} (2007): \enquote{A theory of (relative)
  discounting,} \emph{Journal of Economic Theory}, 137, 214--245.

\bibitem[\protect\citeauthoryear{Rubinstein}{Rubinstein}{2003}]{rubinstein2003economics}
\textsc{Rubinstein, A.} (2003): \enquote{``{E}conomics and psychology?'' The
  case of hyperbolic discounting,} \emph{International Economic Review}, 44,
  1207--1216.

\bibitem[\protect\citeauthoryear{Schmeidler}{Schmeidler}{1989}]{schmeidler89}
\textsc{Schmeidler, D.} (1989): \enquote{Subjective Probability and Expected
  Utility without Additivity,} \emph{Econometrica}, 57, 571--587.

\bibitem[\protect\citeauthoryear{Simon}{Simon}{1955}]{simon1955behavioral}
\textsc{Simon, H.~A.} (1955): \enquote{A behavioral model of rational choice,}
  \emph{The Quarterly Journal of Economics}, 99--118.

\bibitem[\protect\citeauthoryear{Starr}{Starr}{1969}]{starr1969quasi}
\textsc{Starr, R.~M.} (1969): \enquote{Quasi-equilibria in markets with
  non-convex preferences,} \emph{Econometrica}, 25--38.

\bibitem[\protect\citeauthoryear{Thaler}{Thaler}{1981}]{THALER1981201}
\textsc{Thaler, R.} (1981): \enquote{Some empirical evidence on dynamic
  inconsistency,} \emph{Economics Letters}, 8, 201--207.

\bibitem[\protect\citeauthoryear{Varian}{Varian}{1990}]{VARIAN1990125}
\textsc{Varian, H.~R.} (1990): \enquote{Goodness-of-fit in optimizing models,}
  \emph{Journal of Econometrics}, 46, 125--140.

\bibitem[\protect\citeauthoryear{Wu and Gonz\'{a}lez}{Wu and
  Gonz\'{a}lez}{1996}]{wu1996curvature}
\textsc{Wu, G. and R.~Gonz\'{a}lez} (1996): \enquote{Curvature of the
  probability weighting function,} \emph{Management science}, 42, 1676--1690.

\end{thebibliography}

\end{document}